\documentclass[11pt,a4paper]{article}
\pdfoutput=1
\usepackage[margin=1in]{geometry}
\usepackage{amsmath,amssymb,amsthm,booktabs,color,doi,enumitem,graphicx,url}
\usepackage{thmtools, thm-restate}
\usepackage{stmaryrd}
\usepackage[numbers,sort&compress]{natbib}
\usepackage{microtype}
\usepackage{hyperref}

\newtheorem{theorem}{Theorem}
\newtheorem{corollary}[theorem]{Corollary}
\newtheorem{lemma}[theorem]{Lemma}

\theoremstyle{definition}
\newtheorem{definition}{Definition}

\DeclareMathOperator{\poly}{poly}
\DeclareMathOperator{\polylog}{polylog}
\DeclareMathOperator{\dist}{dist}

\hypersetup{
    colorlinks=true,
    linkcolor=black,
    citecolor=black,
    filecolor=black,
    urlcolor=[rgb]{0,0.1,0.5},
    pdfauthor={Laurent Feuilloley, Juho Hirvonen, Jukka Suomela},
    pdftitle={Locally optimal load balancing},
}

\newenvironment{myabstract}
               {\list{}{\listparindent 1.5em%
                        \itemindent    \listparindent
                        \leftmargin    0pt
                        \rightmargin   0pt
                        \parsep        0pt}%
                \item\relax}
               {\endlist}

\newenvironment{mycover}
               {\list{}{\listparindent 0pt
                        \itemindent    \listparindent
                        \leftmargin    0pt
                        \rightmargin   0pt
                        \parsep        0pt}%
                \raggedright
                \item\relax}
               {\endlist}

\begin{document}

\vspace*{2ex}
\begin{mycover}
{\huge \bfseries Locally Optimal Load Balancing}
\bigskip
\bigskip

\textbf{Laurent Feuilloley}

\nolinkurl{laurent.feuilloley@ens-cachan.fr}
\medskip

{\small \'Ecole Normale Sup\'erieure de Cachan, France\par}
\medskip
{\small Helsinki Institute for Information Technology HIIT, \\ Department of Computer Science, Aalto University, Finland\par}

\bigskip

\textbf{Juho Hirvonen}

\nolinkurl{juho.hirvonen@aalto.fi}
\medskip

{\small Helsinki Institute for Information Technology HIIT, \\ Department of Computer Science, Aalto University, Finland\par}
\bigskip

\textbf{Jukka Suomela}

\nolinkurl{jukka.suomela@aalto.fi}
\medskip

{\small Helsinki Institute for Information Technology HIIT, \\ Department of Computer Science, Aalto University, Finland\par}

\end{mycover}
\bigskip
\begin{myabstract}
\noindent\textbf{Abstract.}
This work studies distributed algorithms for \emph{locally optimal load-balancing}: We are given a graph of maximum degree $\Delta$, and each node has up to $L$ units of load. The task is to distribute the load more evenly so that the loads of adjacent nodes differ by at most $1$.

If the graph is a path ($\Delta = 2$), it is easy to solve the \emph{fractional} version of the problem in $O(L)$ communication rounds, independently of the number of nodes. We show that this is tight, and we show that it is possible to solve also the \emph{discrete} version of the problem in $O(L)$ rounds in paths.

For the general case ($\Delta > 2$), we show that fractional load balancing can be solved in $\poly(L,\Delta)$ rounds and discrete load balancing in $f(L,\Delta)$ rounds for some function $f$, independently of the number of nodes.
\end{myabstract}

\thispagestyle{empty}
\setcounter{page}{0}
\newpage

\section{Introduction}

In this work, we introduce the problem of \emph{locally optimal load balancing}, and study it from the perspective of distributed algorithms.

\subsection{Locally optimal load balancing}\label{ssec:intro-def}

In this problem, we are given a graph $G = (V,E)$, and each node has up to $L$ units of load. The task is to distribute load more evenly so that the loads of adjacent nodes differ by at most $1$:
\begin{center}
    \includegraphics[page=1]{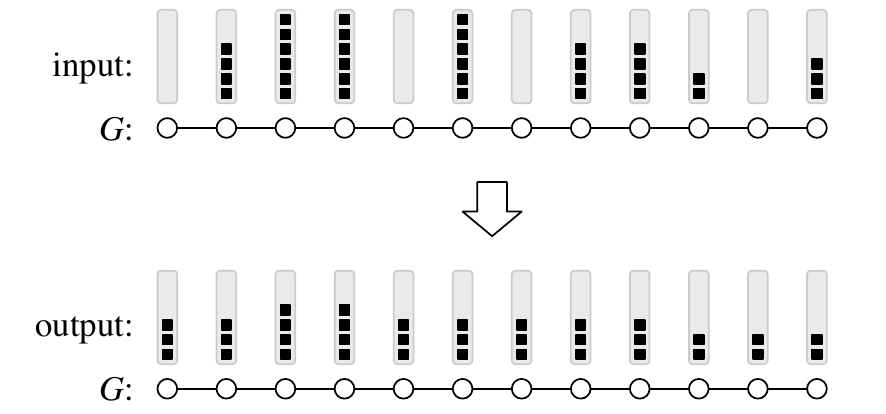}
\end{center}
That is, we want to \emph{smooth out} the load distribution, and find an \emph{equilibrium} in which no edge can improve its load distribution by selfishly moving load between its endpoints.

A bit more formally, in the load balancing problem we are given an input vector $x\colon V \to \{0,1,\dotsc,L\}$, and the task is to find an output vector $y\colon V \to [0,L]$ and a flow $f\colon E \to \mathbb{R}$ so that for each node $v \in V$ we have
\begin{equation}\label{eq:preserve}
    y(v) = x(v) + \sum_{(u,v) \in E} f(u,v),
\end{equation}
and for each edge $(u,v) \in E$ we have
\begin{equation}\label{eq:happy}
    | y(u) - y(v) | \le 1.
\end{equation}
Here is an illustration of the input and a feasible solution in the special case that $G$ is a path:
\begin{center}
    \includegraphics[page=2]{figs.pdf}
\end{center}
The problem comes in two natural flavours:
\begin{itemize}[noitemsep]
    \item \emph{Discrete load balancing}: $y(v) \in \{0,1,\dotsc,L\}$, i.e., load units are indivisible.
    \item \emph{Fractional load balancing}: $y(v) \in [0,L]$, i.e., load units can be divided.
\end{itemize}

\subsection{Centralised algorithms}\label{ssec:intro-central}

Both discrete and fractional load balancing can be solved easily with the following algorithm: Start with $y \gets x$ and $f \gets 0$. Then repeatedly pick an \emph{unhappy edge} $(u,v) \in E$ with $y(u) \ge y(v) + 2$, and move one unit of load from $u$ to $v$. This algorithm clearly converges, as the potential function $\sum_v y(v)^2$ decreases by at least $2$ in each step.

\subsection{Local solutions and local algorithms}

In the above centralised algorithm, we can think that each node $v$ has a pile of $y(v)$ tokens and we always move the topmost token. Then the height of a token decreases by at least one every time we move it; hence no individual token is moved more than $L$ times. This argument shows that there always exists a \emph{local solution} in which the final position of a token is always within distance $L$ from its origin; that is, each token can stay in its radius-$L$ neighbourhood.

In this work we are interested if the problem can be solved with a \emph{local algorithm}: is it possible to solve the problem so that we can compute the flow $f(u,v)$ for each edge $(u,v) \in E$ based on only the information that is available within distance $T$ from $(u,v)$ in graph $G$, for some~$T$. Equivalently, we want to know if there is a (deterministic) distributed algorithm in the usual LOCAL model that solves the load balancing problem in $T$ \emph{communication rounds}, or more succinctly, in \emph{time} $T$.

We will assume that the input graph has maximum degree $\Delta$. We are interested in local algorithms with a running time of $T = T(L,\Delta)$ that may depend on the maximum load $L$ and maximum degree $\Delta$, but is independent on the number of nodes $n = |V|$. Such an algorithm could be used to solve load balancing even in infinitely large graphs, and it would be very easy to e.g.\ parallelise such algorithms, as each part of the output can be determined based on its local neighbourhood.

\subsection{Smoothing with moving average}\label{ssec:intro-avg}

There is a special case that can be easily solved with a local algorithm in time $T = O(L)$: fractional load balancing in $2$-regular graphs (cycles and infinite paths). We can simply calculate the moving average of the input loads with a window of size $\Theta(L)$. More concretely, each node gives a fraction ${1/(2L+1)}$ of its input load to every node (including itself) in its radius-$L$ neighbourhood. This way the final loads of adjacent nodes differ by at most ${L/(2L+1)} < 1/2$ units. The same strategy can be applied easily in, e.g., $d$-dimensional grids.

Among others, the present work seeks to answer the following questions:
\begin{itemize}
    \item Is the running time of $O(L)$ optimal here, or could we solve it in time $o(L)$?
    \item Can we generalise this kind of smoothing algorithms to arbitrary graphs, and if so, what is the running time?
    \item Can we generalise this kind of smoothing algorithms to discrete load balancing?
\end{itemize}

\subsection{Contributions}

The contributions of this work are as follows. We start with a simple lower bound:
\begin{restatable}{theorem}{thmneglinear}\label{thm:neg-linear}
    Load balancing requires $\Omega(L)$ rounds, even in the case of paths and cycles.
\end{restatable}
Then we prove negative results for various algorithm families that have been used widely in the prior work. To this end, we define the following algorithm families:
\begin{itemize}
    \item \emph{Match-and-balance algorithms}: In each step, the algorithm finds a matching $M$ and balances the load (fully or partially) for each edge in $M$. More precisely, for each edge $(u,v) \in M$ with $y(u) > y(v)$, the algorithm increases the flow $f(u,v)$ by at most ${(y(u) - y(v))}/2$. For example, many natural distributed versions of the centralised algorithm from Section~\ref{ssec:intro-central} are of match-and-balance type.
    \item \emph{Careful algorithms}: In each round, for each edge $(u,v) \in E$, the algorithm increases or decreases $f(u,v)$ by at most $\poly(L)$. All match-and-balance algorithms are also careful algorithms.
    \item \emph{Oblivious algorithms}: The total amount of load moved from node $u$ to $v$ only depends on the initial load of $u$ and the distance between $u$ and $v$. For example, the moving average algorithm from Section~\ref{ssec:intro-avg} is oblivious.
\end{itemize}
We show that algorithms of any of these types cannot find a locally optimal load balancing efficiently (or at all):
\begin{restatable}{theorem}{thmnegmab}\label{thm:neg-mab}
    Any match-and-balance algorithm takes $\Omega(L^2)$ rounds in the worst case, even in paths and cycles.
\end{restatable}
\begin{restatable}{theorem}{thmnegcareful}\label{thm:neg-careful}
    Any careful algorithm takes $\Delta^{\Omega(L)}$ rounds in the worst case.
\end{restatable}
\begin{restatable}{theorem}{thmnegoblivious}\label{thm:neg-oblivious}
    There are no oblivious algorithms for infinite $d$-regular trees with $d \ge 3$.
\end{restatable}
We then present the main contributions---local algorithms for load balancing. First, we show that we can circumvent the barrier of Theorem~\ref{thm:neg-mab}:
\begin{restatable}{theorem}{thmpospaths}\label{thm:pos-paths}
    Discrete load balancing can be solved in time $O(L)$ in paths and cycles, with a deterministic local algorithm.
\end{restatable}
\begin{corollary}
    The time complexity of both fractional and discrete load balancing in paths and cycles is $\Theta(L)$.
\end{corollary}
Next we show that we can also circumvent the barriers of Theorem \ref{thm:neg-careful} and \ref{thm:neg-oblivious} for fractional load balancing---naturally, we have to design an algorithm that is neither oblivious nor careful:
\begin{restatable}{theorem}{thmposfractional}\label{thm:pos-fractional}
    Fractional load balancing can be solved in time $\poly(L,\Delta)$ in bounded-degree graphs with a deterministic local algorithm.
\end{restatable}
Finally, we show that discrete load balancing can be solved locally, i.e., in time that is independent of~$n$:
\begin{restatable}{theorem}{thmposdiscrete}\label{thm:pos-discrete}
    Discrete load balancing can be solved in time $T(L,\Delta)$, for some function $T$, in bounded-degree graphs with a deterministic local algorithm.
\end{restatable}
Whether there is an \emph{efficient} algorithm for discrete load balancing in the general case remains an open question.

\section{Related work}

There is a vast body of literature related to problems that are superficially similar to locally optimal load balancing. However, in many cases the primary goal is something else---for example, achieving a near-optimal global solution---and the algorithms just happen to also find a locally optimal solution.

Most of the previous solutions are inefficient. In particular, we are not aware of any solution that comes close to $O(L)$ for discrete load balancing on paths, or close to $\poly(L,\Delta)$ for fractional load balancing in general graphs. In prior work, the inefficiency typically stems from at least one of the following factors:
\begin{enumerate}
    \item \emph{Inherently global problems}: A lot of prior work focuses on problems that are inherently global---for example, the task is to find a solution such that the difference between the minimum load and the maximum load is at most $1$. It is easy to see that any algorithm for solving such problems takes $\Omega(n)$ rounds in the worst case.
    \item \emph{Natural but inefficient algorithms}: Many papers study various natural processes for doing load balancing. Many of these are of match-and-balance type, and virtually all of these are careful. Typically, the negative results of Theorems \ref{thm:neg-mab} and \ref{thm:neg-careful} apply.
\end{enumerate}
In contrast, we study a problem that can be solved efficiently, and our algorithms demonstrate that it is indeed possible to break the barriers of Theorems \ref{thm:neg-mab} and \ref{thm:neg-careful}. In what follows, we will discuss related work in more detail.

\paragraph{Reducing a global potential with local rules.}

There is a lot of literature on load balancing when the goal is to reduce a global potential function by iterating a local balancing rule. Examples of such potential functions are the difference between the maximum and the minimum load (\emph{discrepancy}), the maximum load (\emph{makespan}), and the quadratic difference to the average load.

Various models are considered: two classic models are the \emph{diffusion model}, where vertices distribute their load to all their neighbours, and the \emph{matching model}, where the load is exchanged only along the edges of a matching---for example a random matching or an edge colouring.

In the \emph{continuous case}, where the loads are assumed to be infinitely divisible, the speed of convergence was analysed for simple schemes both in the diffusion model \cite{sinclair89approximate, rabani98local} and the matching model~\cite{ghosh96random, boyd06randomized}. In both the speed of convergence is essentially captured by the spectral properties of the graph in question.

In the context of indivisible loads, known as the \emph{discrete case}, similar problems were first studied for networks designed to balance the load quickly~\cite{peleg89token}. Different schemes for reducing the discrepancy in the discrete case were analysed, the question of whether the speed of convergence in the continuous case could be matched, remained open~\cite{aiello93approximate-load-balancing, ghosh96random, muthukrishnan98first-order, ghosh99tight-analyses}. Recently Sauerwald and Sun~\cite{sauerwald12tight} were able to prove convergence as fast as in the continous case, up to constant factors. Reducing discrepancy is a global problem and can take linear time in the worst case.

\paragraph{Semi-matching problem.}

In the \emph{semi-matching problem} the nodes of a graph are divided into clients and servers~\cite{harvey06semi-matchings}. Each client has to be assigned to an adjacent server. The goal is to optimise the total waiting time of the clients.

Czygrinow et al.~\cite{czygrinow12semimatching} presented a distributed algorithm for finding a locally optimal semi-matching in time $\poly(\Delta)$; this also implies a factor-$2$ approximation of globally optimal semi-matchings.

The semi-matching problem is very similar to the locally optimal load balancing problem, especially when limited to the case of degree 2 clients, with the tokens being more ``localised''. Indeed, our linear lower bound can be adapted to prove an $\Omega(\Delta)$ lower bound for locally optimal semi-matchings.

\paragraph{Balls into bins.}

In the \emph{$d$-choice process} each of $n$ balls goes in the least loaded of $d$ random bins. Dependency of the maximum load on the parameter $d$ is well known~\cite{azar99balanced, karp92efficient, vocking03how-asymmetry}. The choice of the bins can be modelled by a graph~\cite{kenthapadi06balanced}; in one variant the bins are connected by edges and each ball does a local search until it finds a local minimum~\cite{bogdan13balls-local-search, bringmann14balls}. This process produces a locally optimal load balancing.

\paragraph{Sandpile models and chip-firing games.}

Our stability condition is similar to what is used in \emph{sandpile models}~\cite{bak87self-organized,dhar06self-organized,kadanoff89scaling} and \emph{chip-firing games}~\cite{anderson89disks}. However, in these problems the goal is usually to describe final configurations for fixed, very simple algorithms that simulate a natural phenomenon.

\paragraph{Filtering.}

\emph{Sliding window algorithms} for computing the \emph{running average} or for \emph{image filtering} are natural local algorithms. Averaging type algorithms, however, cannot guarantee an integral solution to load balancing problems. \emph{Median filtering} does guarantee integral solutions for integral inputs; however, it does not preserve the total load.

\paragraph{Games and equilibriums.}

The locally optimal load balancing problem can be seen as a problem of finding an equilibrium state, where no single load token can gain advantage by moving. We show that such an equilibrium can be found locally, that is, the decisions made in one part of the graph do not propagate too far. This is in contrast with problems such as finding \emph{stable matchings}, where there is a local algorithm only for finding almost-stable matchings~\cite{floreen10almost-stable}.

\paragraph{Matchings.}

Locally optimal load balancing is closely relate to \emph{bipartite maximal matching}: if the initial loads are $x(v) \in \{0,2\}$, then it is easy to see that a solution can be found using a bipartite maximal matching algorithm. This is a problem that can be solved in time $O(\Delta)$~\cite{hanckowiak98distributed}. Showing a matching lower bounds is a major open question, and we do not expect that one can prove tight lower bounds for locally optimal load balancing as a function of $\Delta$ before we resolve the distributed time complexity of bipartite maximal matching.

In our algorithms for discrete load balancing, we will use the bipartite maximal matching~\cite{hanckowiak98distributed} algorithm as a subroutine. For fractional load balancing, we use the \emph{almost-maximal fractional matching algorithm} due to Khuller et al.~\cite{khuller94primal-dual} as a subroutine.

\section{Negative results}

We will now prove the negative results of Theorems \ref{thm:neg-linear}--\ref{thm:neg-oblivious}. For simplicity, we prove the statements for deterministic distributed algorithms; it is fairly straightforward to extend the results to randomised algorithms (e.g., consider the expected values of the outputs).

Recall that in Section~\ref{ssec:intro-def} we defined the problem so that the output is bounded by $L$. However, we will not exploit this restriction in any of the lower-bound proofs. The negative results hold verbatim for a relaxed version of the problem in which the outputs can be any nonnegative real numbers. We only assume that the inputs are bounded by~$L$.

\subsection{Load balancing on paths and cycles}\label{ssec:neg-linear}

We start with the unconditional lower bound that holds for any algorithm, for both fractional and discrete load balancing, and in the simplest possible case of paths or cycles.

\thmneglinear*

\begin{proof}
We will give the proof for the case of paths; the case of cycles is very similar. Consider a path $P$ with $n$ nodes, labelled with the numbers $1,2,\dotsc,n$ from left to right, for a sufficiently large $n$. Let $A$ be a load-balancing algorithm. For an input $x\colon v \to L$, we write $A(x)$ for the output of $A$ on input $x$. Let $h = \lfloor L/2\rfloor-1$.

Consider the following constant inputs: $x_0\colon v \to 0$ and $x_L\colon v \to L$. Let $y_0 = A(x_0)$ and $y_L = A(x_L)$. Clearly $y_0(v) = 0$ for all $v$ and $y_L(v) \ge L$ for at least one $v$. Hence we can find two nodes, $\ell$ and $r$, such that
\[
    y_0(\ell) = 0, \quad
    y_L(r) \ge L, \quad
    |r-\ell| = L-1.
\]
See Figure~\ref{fig:neg-linear} for an illustration.

\begin{figure}[ht]
    \centering
    \includegraphics[scale=1.2,page=8]{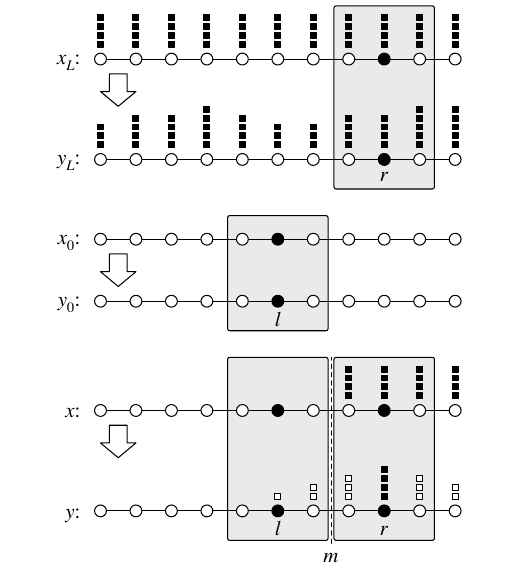}
    \caption{The proof of Theorem~\ref{thm:neg-linear} in Section~\ref{ssec:neg-linear}. In this example, $L = 4$ and $h = 1$. We can find a node $r$ with output at least $L$ in $y_L$, and a node $\ell$ with output $0$ in $y_0$ so that the distance between $\ell$ and $r$ is $L-1$. Then we construct instance $x$ that looks like $x_L$ in the $h$-neighbourhood of $r$ and it looks like $x_0$ in the $h$-neighbourhood of $\ell$. If node $r$ does not change its output between $y_L$ and $y$, then node $\ell$ has to change its output between $y_0$ and $y$. Hence the running time is at least $h+1$.}\label{fig:neg-linear}
\end{figure}

W.l.o.g., assume that $\ell < r$. Let $m = (r+\ell)/2$ be the midpoint between $\ell$ and $r$. Now define an input $x$ such that $x(i) = 0$ for $i \le m$ and $x(i) = L$ otherwise. Note that the radius-$h$ neighbourhoods of $\ell$ are identical in $x_0$ and $x$. Similarly, the radius-$h$ neighbourhoods of $r$ are identical in $x_L$ and $x$.

Let $y = A(x)$. If $y(\ell) = y_0(\ell)$ and $y(r) = y_L(r)$, we have a contradiction: the distance between $\ell$ and $r$ is smaller than their load difference, and hence there has to be an unhappy edge between them. Therefore $y(\ell) \ne y_0(\ell)$ or $y(r) \ne y_L(r)$. In both cases, there is a node $v$ that changed its output between two instances, even though the inputs were identical up to distance~$h$. Hence the running time of $A$ has to be at least $h+1 = \Theta(L)$.
\end{proof}

\subsection{Match-and-balance algorithms}\label{ssec:neg-mab}

Recall that in each round, a match-and-balance algorithm finds some matching $M$, and then for each edge $(u,v) \in M$ with $y(u) > y(v)$, the algorithm increases the flow $f(u,v)$ by at most ${(y(u) - y(v))}/2$. Note that $M$ does not need to be a maximal matching, a maximum matching, or a random matching---the following lower bound holds regardless of how clever the algorithm tries to be in its selection of the matching~$M$, and even if it gets the matchings in zero time from an oracle.

\thmnegmab*

The basic idea of the proof is simple. Let $A$ be a match-and-balance algorithm.
\begin{enumerate}[noitemsep]
    \item We construct an instance in which $A$ has to move $\Omega(L^3)$ units of load in total.
    \item We prove that $A$ can move only $O(L)$ units of load per round.
\end{enumerate}
Hence we have a lower bound of $\Omega(L^2)$ for the running time of $A$.

We will again study the case of paths; the case of cycles is very similar. Let $P$ be a path with $2n+1$ nodes, labelled with $-n,-n+1,\dotsc,n$ from left to right. We say that a load vector is \emph{monotone} if $y(i) \ge y(j)$ for all $i \le j$; see Figure~\ref{fig:mab-monotone}. The key feature of match-and-balance algorithms is that a monotone load vector remains monotone after each step.

\begin{restatable}{lemma}{lemmabmonotone}\label{lem:mab-monotone}
    Match-and-balance algorithms maintain a monotone load configuration on $P$.
\end{restatable}

\begin{figure}
    \centering
    \includegraphics[scale=1.2,page=9]{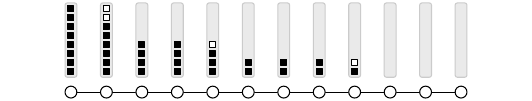}
    \caption{The proof of Lemma~\ref{lem:mab-monotone} in Section~\ref{ssec:neg-mab}. In this example, $L = 8$ and the load distribution is monotone. A match-and-balance algorithm can only move at most $L/2 = 4$ units of load (the highlighted tokens).}\label{fig:mab-monotone}
\end{figure}

\begin{proof}
Assume that the current load configuration $y$ is monotone. Let $M$ be a matching and let $y'$ be the load configuration after balancing over $M$. Consider nodes $i$ and $i+1$. Initially $y(i) \ge y(i+1)$; we will prove by a case analysis that $y'(i) \ge y'(i+1)$:
\begin{enumerate}[noitemsep]
    \item $\{i,i+1\} \in M$: we will have $y'(i) \ge y'(i+1)$.
    \item $\{i,i+1\} \notin M$:
    \begin{itemize}[noitemsep]
        \item $\{i,i-1\} \notin M$: we will have $y'(i) = y(i)$.
        \item $\{i,i-1\} \in M$: we will have $y'(i) \ge y(i)$.
        \item $\{i+1,i+2\} \notin M$: we will have $y'(i+1) = y(i+1)$.
        \item $\{i+1,i+2\} \in M$: we will have $y'(i+1) \le y(i+1)$.
    \end{itemize}
\end{enumerate}
In each case $y'(i) \ge y'(i+1)$.
\end{proof}

In a monotone configuration, we can only move $O(L)$ units of load per round---see Figure~\ref{fig:mab-monotone}.

\begin{lemma}\label{lem:mab-work}
    Any match-and-balance algorithm $A$ can move at most $L/2$ units of load in a single round on path $P$ with a monotone load configuration.
\end{lemma}

\begin{proof}
   Since $A$ maintains a monotone load configuration, the sum of the load differences over all edges is at most $L$. Therefore even if $M$ contains all edges with a non-zero load difference, the algorithm can move only at most $L/2$ units of load per round in total.
\end{proof}

\begin{proof}[Proof of Theorem~\ref{thm:neg-mab}]
We will consider the input vector $x$ where $x(i) = L$ for $i \le 0$ and $x(i) = 0$ otherwise. The vector is monotone and hence it remains monotone throughout the execution of~$A$. Consider the output of node $0$. There are two cases; see Figure~\ref{fig:neg-mab}:
\begin{enumerate}[label=(\alph*)]
    \item The output of node $0$ is at most $h = L/2$. Now for each $i = 0,1,\dotsc,h-1$, we can observe that the load of node $-i$ has decreased by at least $h-i$ units, and by monotonicity, all of this load has been moved to the right. In particular, for each $i$ we have moved $h-i$ units of load from node $-i$ over at least $i+1$ edges. The total amount of work done by the nonpositive nodes is at least the tetrahedral number
    $
        1\cdot h + 2\cdot (h-1) + \dotso + h \cdot 1 = \Theta(h^3) = \Theta(L^3).
    $
    \item The output of node $0$ is at least $h = L/2$. Now for each $i = 0,1,\dotsc,h-1$, we can observe that the load of node $i$ has increased by at least $h-i$ units, and by monotonicity, all of this load has been moved from the left. The total amount of work done by the nonnegative nodes is at least $\Theta(L^3)$.
\end{enumerate}
By Lemma~\ref{lem:mab-work}, moving $\Theta(L^3)$ units of load takes $\Omega(L^2)$ rounds.
\end{proof}

\begin{figure}[ht]
    \centering
    \includegraphics[scale=1.2,page=3]{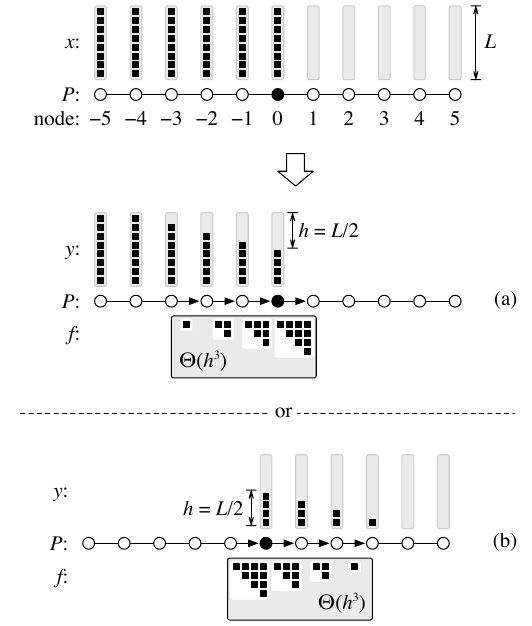}
    \caption{The proof of Theorem~\ref{thm:neg-mab} in Section~\ref{ssec:neg-mab}. We construct input $x$, run any match-and-balance algorithm, and have a case analysis based on the output of node $0$:
    (a)~The load of node $0$ decreases by at least $h$, and the nonpositive nodes do $\Omega(h^3)$ units of work in total to push load to the right.
    (b)~The load of node $0$ is still at least $h$, and the nonnegative nodes do $\Omega(h^3)$ units of work in total to pull load from the left.}\label{fig:neg-mab}
\end{figure}

\subsection{Careful algorithms}\label{ssec:neg-careful}

Recall that careful algorithms move $O(L)$ units of load per round---this includes, for example, all match-and-balance algorithms, as well as many other natural algorithms that simulate the physical process of collapsing piles of tokens.

\thmnegcareful*

\begin{proof}
Construct the input $(G,x)$ as shown in Figure~\ref{fig:neg-careful}: We have a tree $G_u$ rooted at $u$, a tree $G_v$ rooted at $v$, plus an edge $\{u,v\}$. Both trees are of depth $L/4$; each non-leaf node has $d-1$ children. All nodes of $G_u$ have an input load of $0$, and all nodes of $G_v$ have an input load of $L$.

\begin{figure}[ht]
    \centering
    \includegraphics[scale=1.2,page=7]{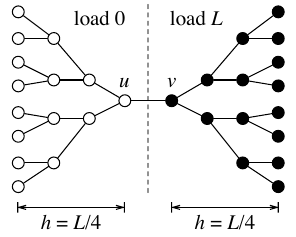}
    \caption{The proof of Theorem~\ref{thm:neg-careful} in Section~\ref{ssec:neg-careful}. In this example, $d = 3$ and $L = 12$.}\label{fig:neg-careful}
\end{figure}

Now consider any solution $(y,f)$. If $y(u) \ge L/4$, then all nodes of $G_u$ have a load of at least $1$, and there are $d^{\Omega(L)}$ nodes in $G_u$. All of the load has been moved across the edge $\{u,v\}$, and hence $f(v,u) = d^{\Omega(L)}$. Otherwise $y(u) < L/4$, and $y(v) < L/4+1$. In this case all nodes of $G_v$ have a load of at most $L-1$, and again we can conclude that $f(v,u) = d^{\Omega(L)}$.

A careful algorithm starts with $y \gets x$ and $f \gets 0$ and changes each element of $f$ by at most $\poly(L)$ in each round. Hence any careful algorithm has to spend $d^{\Omega(L)}$ for this instance.
\end{proof}

\subsection{Oblivious algorithms}\label{ssec:neg-oblivious}

Recall that in an oblivious algorithm, the total amount of load moved from node $u$ to $v$ only depends on the initial load of $u$ and the distance between $u$ and $v$. For example, the algorithm that computes the moving average in an infinite path is an oblivious algorithm. We show that such algorithms do not exist for infinite regular trees of a degree larger than~$2$.

\thmnegoblivious*

\begin{proof}
We say that a node is \emph{full} if it has a load of $L$, and \emph{empty} if it has a load of $0$. We say that a subtree is full if all nodes in it are full, and a subtree is empty if all nodes in it are empty.

Construct the input $(G,x)$ as shown in Figure~\ref{fig:neg-oblivious}a:
\begin{itemize}[noitemsep]
    \item $G$ is the $d$-regular infinite tree,
    \item $\{u,v\}$ is an edge of $G$,
    \item each node $w$ that is closer to $u$ than $v$ is empty,
    \item each node $w$ that is closer to $v$ than $u$ is full.
\end{itemize}
We will consider the infinite tree $G$ rooted at either $u$ or $v$.  If we root it at $u$, then $u$ is adjacent to $1$ full subtree and $d-1$ empty subtrees. If we root it at $v$, then $v$ is adjacent to $d-1$ full subtrees and $1$ empty subtree. See Figure~\ref{fig:neg-oblivious}a for illustrations.

\begin{figure}[ht]
    \centering
    \includegraphics[scale=1.2,page=4]{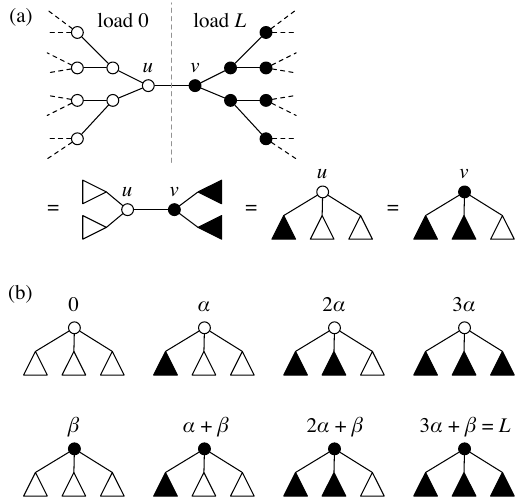}
    \caption{The proof of Theorem~\ref{thm:neg-oblivious} in Section~\ref{ssec:neg-oblivious}. In this example, $d = 3$. (a)~The input graph $G$ is an infinite $d$-regular tree rooted at $\{u,v\}$; one half has an input load of $0$, and the other half has an input load of $L$. (b)~In the output, each full subtree contributes $\alpha$ units of load, and the node itself contributes $\beta$ units of load.}\label{fig:neg-oblivious}
\end{figure}

Let $g(r)$ be the amount of load that the oblivious algorithm moves from a full node $w$ to any node that is at distance $r$ from $w$. Define the shorthand notation
\[
    \alpha = \sum_{r=0}^{\infty} (d-1)^r g(r+1),
\]
which has two equivalent interpretations in rooted infinite $d$-regular trees:
\begin{itemize}[noitemsep]
    \item a full root node sends in total $\alpha$ units of load to each subtree,
    \item the root node receives $\alpha$ units from each full subtree.
\end{itemize}
See Figure~\ref{fig:neg-oblivious}b. In total, a full node gives $d\alpha$ units of load to other nodes, so it leaves
\[
    \beta = L - d\alpha
\]
units of load for itself. It is easy to verify that we must have $\beta \ge 0$ and hence $\alpha \le L/d$; otherwise there would be inputs with negative outputs.

Now we are ready to put the pieces together. Node $u$ receives $\alpha$ units of load from its only full subtree, while $v$ receives $(d-1)\alpha$ units of load from its $d-1$ full subtrees; moreover, $v$ leaves $\beta$ units of load for itself. The load difference between $u$ and $v$ is therefore at least
\[
    (d-1)\alpha + \beta - \alpha = L - 2\alpha \ge \frac{d-2}{d} L.
\]
Hence for any $d \ge 3$ and for a sufficiently large $L$, edge $\{u,v\}$ will be unhappy in $G$, no matter which oblivious algorithm we apply.
\end{proof}

\section{Positive results}

We will now prove the positive results of Theorems \ref{thm:pos-paths}, \ref{thm:pos-fractional}, and \ref{thm:pos-discrete}.

\subsection{Discrete load balancing in paths and cycles}
\label{ssec:pos-paths}

We first give an algorithm that exactly matches the lower bound of Theorem~\ref{thm:neg-linear}.

\thmpospaths*

\paragraph{Infinite directed paths.}

We will first show how to do load balancing in an \emph{infinite path with a consistent orientation}. That is, each node $v$ has a degree of $2$, and it can refer to its \emph{left neighbour} $v-1$ and \emph{right neighbour} $v+1$ in a globally consistent manner.

We will interpret the path with tokens as a $2$-dimensional grid, indexed by $(v,i)$, where $v \in V$ is a node and $i \in \{1,\dotsc,L\}$ is a possible location for a token. We say that $(v,i)$ is a \emph{slot}. Initially, slot $(v,i)$ holds a \emph{token} if $x(v) \ge i$. Our plan is to move the tokens around in the grid so that we maintain the following stability conditions---see Figure~\ref{fig:slots} for an illustration.

\begin{figure}[ht]
    \centering
    \includegraphics[scale=1.2,page=10]{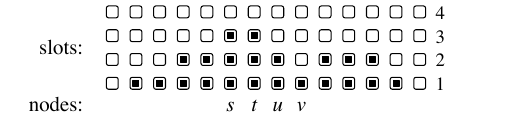}
    \caption{Stability. Token $(s,3)$ is $2$-stable, as there exists a token in slot $(u,2)$, where $u = s+2$, i.e., the node $2$ steps right from $s$. Also $(s,2)$ and $(s,1)$ are $2$-stable. However, this configuration is not $2$-stable: token $(t,3)$ is not $2$-stable, as there is an empty slot $(v,2)$. It can be verified that the configuration is $0$-stable, $1$-stable, and $(-1)$-stable.}\label{fig:slots}
\end{figure}

\begin{definition}
    A token in slot $(v,i)$ is $k$-stable if $i = 1$ or there is a token in slot $(v+k,i-1)$. A configuration is $k$-stable if all tokens are $k$-stable. For a set $K$, a configuration is $K$-stable if it is $k$-stable for all $k \in K$.
\end{definition}

We write $\llbracket a,b\rrbracket = \{a, a+1, \dotsc, b\}$. Initially, the configuration is $0$-stable. If we can find a $\llbracket -1,1\rrbracket$-stable configuration, we can construct a feasible solution to the load balancing problem by simply setting $y(v)$ to be equal to the number of tokens in slots $(v,\cdot)$.

However, we will now design an $O(L)$-time algorithm with a \emph{stronger} stability condition: it will compute a $\llbracket-3,3\rrbracket$-stable configuration. Informally, we smooth out the load distribution so that the slope of the load curve is at most $1/3$. This extra slack will be helpful when we eventually want to solve the problem in paths without consistent orientations.

This algorithm is based on the concept of \emph{pushes}. For a node $v$ and integer $\ell$, define the $\ell$-diagonal of $v$ as the following list of slots (see Figure~\ref{fig:push}):
\[
    S(v,\ell) = \bigl( (v-\ell,1), (v-2\ell,2), \dotsc, (v-L\ell,L) \bigr)
\]
In an $\ell$-push we redistribute the tokens in each $S(v,\ell)$: if there are $k$ tokens in $S(v,\ell)$, then we redistribute the tokens so that the first $k$ elements of $S(v,\ell)$ are occupied and the remaining $L-k$ elements are empty (see Figure~\ref{fig:push}). In essence, we let the tokens slide along each diagonal so that they are piled on the bottom of each diagonal.

\begin{figure}[ht]
    \centering
    \includegraphics[scale=1.2,page=11]{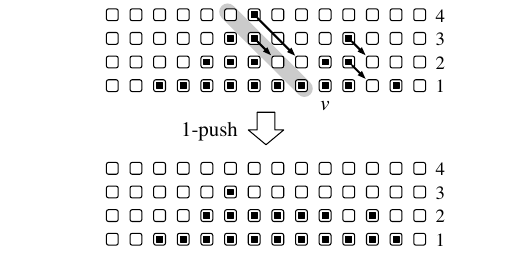}
    \caption{Pushing. The $1$-diagonal of $v$ is highlighted. In a $1$-push, we redistribute the tokens in each $1$-diagonal so the end result will be a $1$-stable configuration. Note that this configuration was already $0$-stable and $(-1)$-stable, and it remained $0$-stable and $(-1)$-stable after a $1$-push. In general, whatever stability we have already achieved by pushing is never lost in subsequent pushes.}\label{fig:push}
\end{figure}

An $\ell$-push can be efficiently implemented in time $O(\ell L)$ with a distributed algorithm: for example, node $v$ is responsible for redistributing the tokens in slots $S(v,\ell)$, and we first use $O(\ell L)$ rounds so that each node $v$ can discover everything related to $S(v,\ell)$, and then another $O(\ell L)$ rounds so that node $v$ can inform the relevant nodes regarding how to move tokens in $S(v,\ell)$.

Clearly, after an $\ell$-push we will have an $\ell$-stable configuration. The non-trivial part is that $\ell$-pushes do not interfere with any stability that we have previously achieved.

\begin{restatable}{lemma}{lempush}\label{lem:push}
For every choice of integers $\ell$ and $k$, if a configuration is $k$-stable, then it is still $k$-stable after an $\ell$-push.
\end{restatable}

\begin{proof}
The case $k = \ell$ is trivial; hence we assume that $k \ne \ell$. Consider slots $a = (v,i)$ and $b = (v+k,i-1)$; see Figure~\ref{fig:pushlemma}. We need to argue that if $a$ holds a token after an $\ell$-push, then $b$ will also hold a token after an $\ell$-push. To this end, let $A$ be the $\ell$-diagonal that contains $a$, and let $B$ be the $\ell$-diagonal that contains $b$. Now by definition, after an $\ell$-push, slot $a$ is occupied if and only if there were at least $i$ tokens in $A$.

\begin{figure}[ht]
    \centering
    \includegraphics[scale=1.2,page=12]{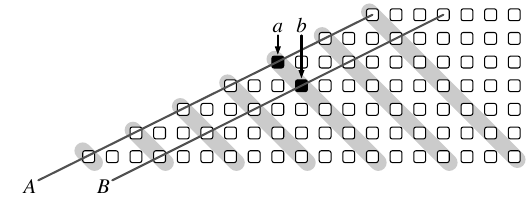}
    \caption{The proof of Lemma~\ref{lem:push}, for $\ell = -2$ and $k = 1$. The configuration was $1$-stable, i.e., the gray $1$-diagonals filled starting from the bottom. Now we do a $(-2)$-push, and want to argue that the configuration will be still $1$-stable. Consider slot $a$. It will be filled iff there are at least $5$ tokens in the $(-2)$-diagonal $A$. But this implies that there are at least $4$ tokens in the diagonal $B$, and hence the slot $b$ will be filled, too.}\label{fig:pushlemma}
\end{figure}

The key observation is that $k$-stability implies that for every token in $A$ there is a token in $B$, with the exception of the first token---if $(u,j) \in A$ holds a token and $j > 1$ then $(u+k,j-1) \in B$ holds a token as well. In particular, if there were at least $i$ tokens in $A$, there were at least $i-1$ tokens in $B$, and hence $b$ will also hold a token.
\end{proof}

Now we can easily find a $\llbracket-3,3\rrbracket$-stable configuration in time $O(L)$: the algorithm simply does an $\ell$-push for each $\ell \in \llbracket-3,3\rrbracket$, sequentially, in an arbitrary order. We will call this algorithm~$A_1$.

\paragraph{Finite directed paths and cycles.}

Algorithm $A_1$ finds a $\llbracket-3,3\rrbracket$-stable configuration in infinite directed paths in time $O(L)$. To handle \emph{finite} directed paths we could extend the algorithm and its analysis so that it takes into account the boundary effects. However, this would be a bit boring---instead, we will show that we can simply take $A_1$ and use it as a black box.

Let us first adjust the stability condition so that it makes sense on finite paths: a token $(v,i)$ is considered $k$-stable also if node $v+k$ does not exist.

Let $T_1 = \Theta(L)$ be the worst-case running time of $A_1$. We use it to construct an algorithm $A_2$ that finds a $\llbracket-3,3\rrbracket$-stable configuration for a finite path $G$, as follows:
\begin{enumerate}
    \item Check if the path is of length at most $4T_1$; if so, we solve the problem by brute force in time $O(L)$, and stop.
    \item Each endpoint $u$ gathers all tokens up to distance $2T_1$ and redistributes them so that all nodes within distance at most $T_1$ from $u$ have the same constant load; let us denote this constant $c(u)$.
    \item Construct a virtual graph $G'$ as follows: each endpoint $u$ pretends that the path continues with infinitely many additional dummy nodes, each with the same constant load $c(u)$.
    \item Simulate algorithm $A_1$ in the virtual graph $G'$.
    \item Discard the dummy nodes.
\end{enumerate}
It is easy to verify that the $A_1$ will never move any tokens across an endpoint, as its neighbourhood was already well-balanced. Therefore if we remove the dummy nodes, we have a feasible solution for $G$. Moreover, the running time of $A_2$ is still $O(L)$.

It is also easy to see that $A_2$ works correctly in directed \emph{cycles}; the first three steps simply do nothing as there are no endpoints.

\paragraph{Undirected paths and cycles.}

So far we have designed an algorithm $A_2$ that finds a $\llbracket-3,3\rrbracket$-stable configuration in paths and cycles with a globally consistent orientation. Now we show how to use it to design an algorithm $A_3$ that finds a $\llbracket-1,1\rrbracket$-stable configuration in paths and cycles without an orientation.

It can be shown that \emph{some} form of local symmetry-breaking is needed. We will use the familiar \emph{port-numbering model}: Each node $v$ has up to two communication ports, labelled with $(v,1)$ and $(v,2)$. The ports are identified with the endpoints of the edges; each edge joins a pair of ports. The port numbers at the endpoints of an edge do not need to match---for example, an edge $\{u,v\}$ may join $(u,1)$ to $(v,1)$ or $(u,1)$ to $(v,2)$.

In algorithm $A_2$, we construct a \emph{virtual graph} $G'$ as shown in Figure~\ref{fig:virtpath}: Each node $v$ splits itself in two virtual nodes, $v_1$ and $v_2$. The virtual nodes also have two ports. For each edge $e = \{u,v\}$, depending on the type of $e$ we connect the virtual nodes of $u$ and $v$ as follows:
\begin{itemize}[noitemsep]
    \item $e$ joins $(u,1)$ to $(v,1)$: connect $(u_1,1)$ to $(v_2,2)$ and $(u_2,2)$ to $(v_1,1)$,
    \item $e$ joins $(u,1)$ to $(v,2)$: connect $(u_1,1)$ to $(v_1,2)$ and $(u_2,2)$ to $(v_2,1)$,
    \item $e$ joins $(u,2)$ to $(v,1)$: connect $(u_1,2)$ to $(v_2,1)$ and $(u_2,1)$ to $(v_1,2)$,
    \item $e$ joins $(u,2)$ to $(v,2)$: connect $(u_1,2)$ to $(v_1,1)$ and $(u_2,1)$ to $(v_2,2)$.
\end{itemize}
If $G$ was a path with $n$ nodes, then $G'$ consists of two disjoint paths with $n$ nodes each. If $G$ was an $n$-cycle, then $G'$ consists of either one cycle with $2n$ nodes or two cycles with $n$ nodes each.

\begin{figure}[ht]
    \centering
    \includegraphics[scale=1.2,page=13]{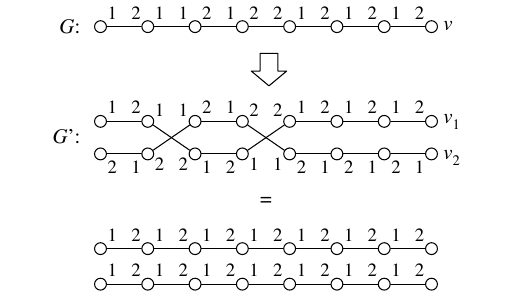}
    \caption{Given any path $G$ with some port numbering, we can construct a virtual graph $G'$ that consists of two paths, both of which have a \emph{consistent} port numbering.}\label{fig:virtpath}
\end{figure}

The key observation is that there is a \emph{consistent} port numbering in $G$: port $1$ of a virtual node is always connected to port $2$ of an adjacent virtual node. We can now interpret the ports so that in each virtual node port $1$ points ``left'' and port $2$ points ``right''.

Each node first splits its input load arbitrarily between its virtual copies. Then we run algorithm $A_2$ to find a $\llbracket-3,3\rrbracket$-stable configuration in the virtual graph, and then map all tokens back to the original graph: the new load of $v$ is the sum of the new loads of $v_1$ and $v_2$; see Figure~\ref{fig:pathfix}.

\begin{figure}[ht]
    \centering
    \includegraphics[scale=1.2,page=14]{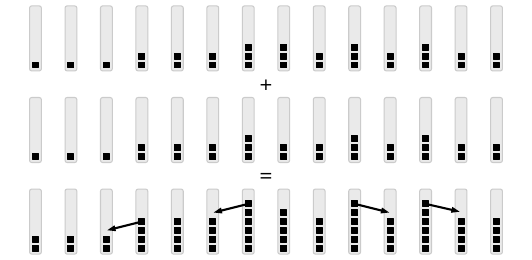}
    \caption{The sum of two $\llbracket-3,3\rrbracket$-stable configurations can be easily turned into a $\llbracket-1,1\rrbracket$-stable configuration with local modifications.}\label{fig:pathfix}
\end{figure}

Now we have a configuration where the maximum load difference between a pair of adjacent nodes is $2$. However, the load is \emph{approximately well-balanced}: a load difference of more than $2$ implies a distance of at least $4$. Therefore we can easily find a $\llbracket-1,1\rrbracket$-stable configuration in $O(1)$ time with local operations (see Figure~\ref{fig:pathfix}). For example, we can apply a match-and-balance algorithm: find a maximal matching $M$ of unhappy edges and move a token over each edge. Conveniently, all edges become happy, including those that were not in~$M$. It is easy to find a maximal matching $M$ in $O(1)$ time, as this is in essence maximal matching in a bipartite graph of maximum degree $2$: on one side we have the nodes that are ``too low'' and on the other side we have the nodes that are ``too high'' in comparison with their neighbours.

In summary, we can find a $\llbracket-1,1\rrbracket$-stable configuration in any path or cycle in time $O(L)$, and therefore we can do discrete load balancing in any path or cycle in time $O(L)$.

\subsection{Discrete load balancing in general graphs}\label{ssec:pos-discrete}

We will now show how to do discrete load balancing in graphs of maximum degree $\Delta$.

\thmposdiscrete*

Again, we will imagine that each node $v$ has $L$ \emph{slots}, labelled $(v,\cdot)$, and each token is placed in one of the slots. Initially slots $(v,1), (v,2), \dotsc, (v,x(v))$ are occupied with tokens.

We define the \emph{(downward) cone} $C(v,i)$ of slot $(v,i)$ as the set of slots $(u, j) \neq (v, i)$ such that $i-j \ge \dist(v, u)$; see Figure~\ref{fig:cone}. In the algorithm, if there is a token in $(v,i)$ and all slots of the cone $C(v,i)$ are full, then we say that the token is \emph{stable}, and we \emph{freeze} it, i.e. it will never be moved again.

\begin{figure}[ht]
    \centering
    \includegraphics[scale=1.2,page=15]{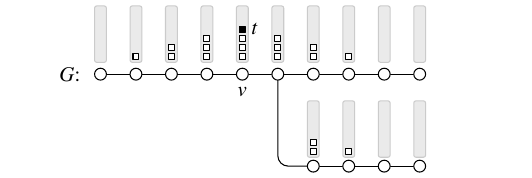}
    \caption{The downward cone $C(v,4)$ of the token $t = (v,4)$ consists of the slots denoted with white boxes.}\label{fig:cone}
\end{figure}

In the algorithm we try to match the highest unfrozen tokens with the free slots in their cones. If they succeed then they move to these slots; otherwise they can be frozen.

We now give the pseudo-code of the algorithm in a centralised way, prove the correctness of the algorithm, and then show that it is actually a local algorithm. The algorithm proceeds as follows:
\begin{enumerate}
	\item All stable tokens of the initial configuration are frozen.
	\item For each $h = L, L-1,\dotsc,1$:
	\begin{enumerate}
		\item Construct the virtual bipartite graph $F_h = (T \cup S, E)$, where $T$ consists of unfrozen tokens at level $h$, $S$ consists of all empty slots at levels below $h$, and there is an edge $\{t,s\}$ if $s \in S$ is an empty slot in the cone of token $t \in T$.
		\item In $F_h$, find a maximal matching $M$.
		\item For every unfrozen token $t$ at level $h$: if the token is matched with a slot $s$ in $M$, move the token to slot $s$, otherwise freeze it. 
		\item Collapse the tokens so that for each node $v$ that holds $k$ tokens, the tokens are in the slots $(v,1), (v,2), \dotsc, (v,k)$.
	\end{enumerate}
\end{enumerate} 

First, remark that we maintain the invariant that at round $h$, all load in slots at height~$h$ either moves down or is safely frozen.  Indeed, if a token is not matched, then all slots in its cone will be full at the end of the loop, and if it is matched, it moves to a strictly lower level, thereafter the invariant is true for level $h$ and maintained for the levels above. At the end of the algorithm all the tokens are frozen, thus the configuration is stable.

We stated the algorithm in a centralised manner, but it is actually local: The vertices only need the knowledge of their radius-$L$ neighbourhood to find their neighbours in graph $F_h$. Graph $F_h$ has a maximum degree of $O(L\Delta^L)$. Therefore we can find a maximal matching in $F_h$ by simulating $O(L\Delta^L)$ rounds of the proposal algorithm \cite{hanckowiak98distributed} in the virtual graph $F_h$. The simulation has a multiplicative $O(L)$ overhead---adjacent nodes in $F_h$ are at distance $O(L)$ in graph $G$. Finally, we have $O(L)$ iterations, giving the overall complexity of $O(L^3 \Delta^L)$.

\subsection{Fractional load balancing in general graphs}\label{ssec:pos-fractional}

In fractional load balancing, we can use the same basic idea as what we had in the discrete case, but much faster:

\thmposfractional*

Our algorithm follows the same basic structure as the discrete algorithm of Section~\ref{ssec:pos-discrete}. However, in each bipartite virtual graph $F_h$, we compute an $\varepsilon$-maximal fractional matching. With the algorithm by Khuller et al.~\cite{khuller94primal-dual}, this can be done in $O(\log \frac{1}{\varepsilon} + \log \Delta)$ rounds, which gives us an exponential speedup over the $O(\Delta)$-round algorithm for maximal bipartite matching.

\paragraph{Almost maximal fractional matchings.}

Let the bipartite graph be $G = (V, E)$ with $V = T \cup S$ and $\Delta$ be the maximum degree. Each node has a \emph{capacity} $c \colon V \to [0,1]$. A fractional matching is a function $y\colon E \to [0,1]$ such that for each node, the sum $y[v] = \sum_{e \ni v} y(e)$ is at most $c(v)$. A fractional matching is \emph{$\varepsilon$-maximal} if
\[
\max_{e \in E} \min \{ c(v) - y[v] : v \in e \} \le \varepsilon,
\]
that is, there is no edge $e$ with a value $y(e)$ that could be increased by more than $\varepsilon$.

\paragraph{Algorithm for fractional load balancing.}

Next we describe the algorithm for finding a fractional load balancing. As before, we have \emph{slots} labelled with $(v,i)$; here $v$ is a node and $i$ is the \emph{level} of the slot. However, now each slot may contain fractional units of load. We adapt the definition of stability to fractional load balancing in a natural manner: we say that $\alpha$ units of the load in slot $(v,i)$ is stable if each $(u,j)$ with $i-j=d(v,u)$ has at least $\alpha$ units of load, and each $(u,j)$ with $i-j>d(v,u)$ is full. In the algorithm we will \emph{freeze} some parts of the load. We use $\ell_i(v)$ to denote the total amount of load in slot $(v,i)$, and $f_i(v)$ to denote the amount of frozen load in slot $(v,i)$. 

The algorithm would be simpler to analyse if we had a maximal fractional matching algorithm, but we only have an efficient $\varepsilon$-maximal one. Our strategy is to round the load, and by doing it, we accumulate \emph{surplus} and \emph{deficit}. This way we can analyse easily the iterations of the algorithm. We keep track of surplus and deficit, and at the end of the algorithm we readjust the loads.

The algorithm for fractional load balancing works as follows. First, double all load; this way we have $2L$ slots per node. Based on the new input, freeze all stable load. Then, for each $h = 2L, 2L-1,\dotsc,1$ perform the following steps:
\begin{enumerate}
    \item Construct the virtual bipartite graph $F_h = (T \cup S, E)$, where $T$ consists of slots with unfrozen load at level $h$, $S$ consists of all slots at levels below $h$, and there is an edge $\{t,s\}$ if $s \in S$ is a non-full slot in the cone of slot $t \in T$.
    \item Define the capacities of $F_h$ as follows: the capacity of a node $(t, h) \in T$ is its unfrozen load $c(t, h) = \ell_h(s) - f_h(t)$, and the capacity of a node $(s,i) \in S$ is its free space $c(s,i) = 1-\ell_i(s)$.
    \item Find an $\varepsilon$-maximal fractional matching $y$ in $F_h$.
    \item\label{step:move} For each edge $e = \{t,s\}$ of $F_h$, move $y(e)$ units of load from slot $t$ to slot $s$.
    \item\label{step:freeze} For each $t$ consider the two cases:
    \begin{itemize}
        \item If $t$ has load at most $\varepsilon$, round it to 0 and freeze the load in $t$. We create at most $\varepsilon$ units of deficit in slot $t$ at the moment we freeze it.
        \item If $t$ has strictly more than $\varepsilon$ load, then each unfrozen slot in the cone of $t$ has at least $1-\varepsilon$ load. We round them to 1 and freeze all load at $t$ and in $C(t)$. We create at most $\varepsilon$ units of surplus in each slot of $C(t)$ at the moment we freeze it.
    \end{itemize}
    \item Collapse all unfrozen load as low as possible.
\end{enumerate}
Finally, undo the rounding---remove the surplus and put back the deficit. Then normalise the output by dividing all load by two.

\paragraph{Analysis.}

Thanks to the rounding, the configuration after each iteration satisfies the same invariant as the discrete algorithm: at round $h$, all load in slots at height $h$ either moves down or is safely frozen. Then the configuration at the end of the iteration phase is stable and for each edge $\{u,v\}$ we have $|y(u)-y(v)|\le 1$. 

Each slot is involved at most once in the rounding, precisely at the moment we freeze it. Hence before normalisation, the difference between the real load and the rounded load is in $[-2L\varepsilon,2L\varepsilon]$ for each node. Therefore we have $|y(u)-y(v)|\le 1+4L\varepsilon$ for each edge $\{u,v\}$ before normalisation and $|y(u)-y(v)|\le 1/2+2L\varepsilon$ after normalisation. We guarantee a feasible solution by choosing $\varepsilon \le 1/(4L)$.

Khuller et al.~\cite{khuller94primal-dual} show how to find an $\varepsilon$-maximal fractional matching in time $O(\log \varepsilon^{-1} + \log d)$ in graphs of maximum degree $d$. The virtual graph $F_h$ has a maximum degree of $d = O(L\Delta^L)$, and there is an $O(L)$ overhead in the simulation of $F_h$ in $G$. Finally, we have $L$ iterations in the algorithm; in total, the running time can be bounded by
\[
O(L^2(\log \varepsilon^{-1} + \log d))
= O(L^3 \log \Delta).
\] 
This completes the proof of Theorem~\ref{thm:pos-fractional}.

\section{Conclusions}

In this work, we have introduced the problem of finding a \emph{locally optimal load balancing}, and studied its distributed time complexity. We have shown that the problem can be solved in a strictly local fashion, but to do it, one has to resort to algorithms that are very different from typical load-balancing strategies that are used in the literature. Among the key findings are:
\begin{itemize}[noitemsep]
    \item an $O(L)$-time algorithms for discrete load balancing in paths and cycles,
    \item a $\poly(L,\Delta)$-time algorithm for fractional load balancing in graphs of maximum degree $\Delta$.
\end{itemize}

The main open question is the distributed time complexity of the discrete load balancing problem. Our algorithm is local, but it has a running time exponential in $L$; the key question is whether $\poly(L,\Delta)$-time algorithms exist. We suspect that it is related to another long-standing open question---the distributed time complexity of bipartite maximal matching. Indeed, a $\polylog(\Delta)$-time algorithm for bipartite maximal matching would imply a $\poly(L,\Delta)$-time algorithm for discrete load balancing. We conjecture that such algorithms do not exist, but proving such lower bounds seems to be still beyond the reach of current techniques.

Another open question is the generalisation of the results from the LOCAL model to the CONGEST model. In particular, the polynomial-time algorithm for fractional load balancing heavily abuses the unlimited bandwidth of the LOCAL model, but it seems that there are no major obstacles for designing an analogous algorithm that works efficiently in the CONGEST model.

\section*{Acknowledgements}

We have discussed this problem and its variants over the years with numerous people, including, at least, 
Sebastian Brandt,
Pierre Fraigniaud,
Mika G\"o\"os,
Petteri Kaski,
Barbara Keller,
Janne H.\ Korhonen,
Juhana Laurinharju,
Tuomo Lempi\"ainen,
Christoph Lenzen,
Joseph S.\ B.\ Mitchell,
Pekka Orponen,
Joel Rybicki,
Thomas Sauerwald,
Stefan Schmid,
and
Jara Uitto.
Many thanks to all of you for your comments!
Computer resources were provided by the Aalto University School of Science ``Science-IT'' project.

\clearpage

\bibliographystyle{plainnat}
\bibliography{articles}

\clearpage

\end{document}